\documentclass[11pt,letterpaper]{article}
\usepackage[margin=1in]{geometry}
\usepackage{amsmath,amsthm,amsfonts,amssymb,amscd}
\usepackage{lastpage}
\usepackage{enumerate}
\usepackage{fancyhdr}
\usepackage{mathrsfs}
\usepackage{xcolor}
\usepackage{graphicx}
\usepackage{listings}
\usepackage[colorlinks=true, urlcolor=blue, linkcolor=blue, citecolor=magenta]{hyperref}
\usepackage{cite}
\usepackage{multirow}
\usepackage{comment}
\usepackage{indentfirst}
\usepackage{hyperref}
\usepackage{cleveref}

\hypersetup{%
  colorlinks=true,
  linkcolor=red,
  linkbordercolor={0 0 1}
}

\crefformat{equation}{(#2#1#3)}

\lstdefinestyle{Matlab}{
    language        = matlab,
    frame           = lines, 
    basicstyle      = \footnotesize,
    keywordstyle    = \color{blue},
    stringstyle     = \color{green},
    commentstyle    = \color{red}\ttfamily
}
\newtheorem{theorem}{Theorem}
\newtheorem{definition}[theorem]{Definition}

\newtheorem{observation}[theorem]{Observation}

\newtheorem{remark}[theorem]{Remark}
\newtheorem{lemma}[theorem]{Lemma}

\newtheorem{fact}[theorem]{Fact}
\newtheorem{proposition}[theorem]{Proposition}

\begin{document}

\title{PCP-free APX-Hardness of Nearest Codeword \\ and Minimum Distance \\
{\normalfont \normalsize  \emph{\textcolor{black}{in memory of 
Alexander Vardy}}}    
}

\author{
Vijay Bhattiprolu
\thanks{Department of Combinatorics \& Optimization, University of Waterloo}
\and
Venkatesan Guruswami
\thanks{Simons Institute for the Theory of Computing and Departments of EECS and Mathematics, UC Berkeley}
\and 
Xuandi Ren
\thanks{Department of EECS, UC Berkeley}
}

\date{}
\maketitle
\begin{abstract}
 We give simple deterministic reductions demonstrating the NP-hardness of approximating the nearest codeword problem and minimum distance problem within arbitrary constant factors (and almost-polynomial factors assuming NP cannot be solved in quasipolynomial time). The starting point is a simple NP-hardness result without a gap, and is thus ``PCP-free." Our approach is inspired by that of Bhattiprolu and Lee~\cite{BL24} who give a 
 PCP-free randomized reduction for similar problems over the integers and the reals.
 We leverage the existence of $\varepsilon$-balanced codes to derandomize and further simplify their reduction for the case of finite fields. 
\end{abstract}

\section{Introduction}

Error-correcting codes are fundamental mathematical objects with many applications across computing. A long line of work has been studying the complexity of basic computational problems on codes, especially the Nearest Codeword Problem (NCP) and the Minimum Distance Problem (MDP)~\cite{BMvT78,ABSS97,Var97,DMS93,CW12,AK14, M14} defined below.

\begin{definition}[Nearest Codeword Problem (NCP)\footnote{NCP is sometimes referred to as the Maximum Likelihood Decoding (MLD) problem in the literature.}]~\\
    An instance of NCP over a finite field $\mathbb F_q$ consists of an affine subspace $V \subseteq \mathbb F_q^{n}$, where $V$ is encoded using any basis. The goal is to find the minimum Hamming weight of any vector $x \in V$.
\end{definition}

\noindent
MDP is a homogeneous variant of NCP: 
\begin{definition}[Minimum Distance Problem (MDP)]~\\
    An instance of MDP over a finite field $\mathbb F_q$ consists of a  subspace $V \subseteq \mathbb F_q^{n}$. The goal is to find the minimum Hamming weight of any nonzero vector $x \in V$. 
\end{definition}

\smallskip\noindent\textbf{Our Results.}
We give a new elementary proof of the inapproximability of these problems, avoiding the PCP theorem. 
It has traditionally been much easier to show hardness for the non-homogeneous variant NCP compared to the homogeneous version MDP. 
Previous approaches first establish inapproximability of NCP as a quick corollary of the PCP theorem, and then employ non-trivial coding gadgets to give a gap-preserving reduction from NCP to MDP. 
However, in this work, we proceed in reverse. We first give a deterministic gap-producing reduction from homogeneous quadratic equations to MDP (using an appropriately chosen tensor product of a concatenated code) and then amplify the gap by simple tensoring. We then give an elementary gap-preserving reduction from MDP to NCP, utilizing a special ``distinguished-coordinate'' property in our hard MDP instance. As a result, we have the following two theorems:

\begin{theorem}
\label{thm:mdp}
For any finite field $\mathbb F_q$, it is NP-Hard to approximate MDP within any constant factor, and for any fixed $\varepsilon>0$, it is Quasi-NP-Hard to approximate MDP within $2^{\log^{1-\varepsilon}n}$.
\end{theorem}
\noindent
A problem is said to be Quasi-NP-Hard if it does not admit a polynomial time algorithm assuming NP cannot be solved in quasipolynomial time. 

\begin{theorem}
\label{thm:ncp}
For any finite field $\mathbb F_q$, it is NP-Hard to approximate NCP within any constant factor, and for any fixed $\varepsilon>0$, it is Quasi-NP-Hard to approximate NCP within $2^{\log^{1-\varepsilon}n}$.
\end{theorem}

\paragraph{Related Work.}
Proving NP-Hardness of exactly computing MDP was a question posed in \cite{BMvT78}, 
and was resolved in the affirmative by Vardy~\cite{Var97}, using as a gadget  
a Reed-Solomon code concatenated with the Hadamard code. The analog for MDP for point lattices, the Shortest Vector Problem (SVP), was shown to be NP-hard under randomized reductions by Ajtai~\cite{Ajtai98}.\footnote{The NP-hardness of SVP under deterministic reductions still remains open.} Micciancio~\cite{Micciancio01} then established hardness of approximating SVP within a constant factor bounded away from $1$ via a reduction from the affine version, the Closest Vector Problem (CVP), using ``locally dense'' lattices as key gadgets.
Dumer, Micciancio and Sudan~\cite{DMS93} adapted the locally dense gadgets framework to the setting of linear codes and established constant-factor inapproximability of MDP over any finite field, under randomized reductions (the inapproximability factor can be amplified to any desired constant via tensoring).
A locally dense code is a code of minimum distance $d$, that admits a Hamming ball of 
smaller radius ($(1-\varepsilon)d$ for some constant $\varepsilon>0$) containing exponentially many codewords. 

Cheng and Wan~\cite{CW12} derandomized the proof of \cite{DMS93} by giving a deterministic construction of a locally dense gadget over any finite field. 
In particular they give an explicit Hamming ball of radius $0.67d$ that contains exponentially many codewords of an explicit code of distance $d$.
The proof of Cheng and Wan's result is fairly deep, making use of Weil's character sum estimate. 

Austrin and Khot~\cite{AK14} gave a simpler proof of deterministic hardness of MDP by making use of tensor codes. 
Building on this, Micciancio~\cite{M14} proved that for the special case of $\mathbb{F}_2$, the tensoring of any base code with large distance yields a significantly simpler deterministic construction of a locally dense gadget than the proof of \cite{CW12}. 

More recently, Bhattiprolu and Lee~\cite{BL24} observed that if the base code $C$ being tensored has a weak version of the local density property (i.e. there is a Hamming ball of radius $(1+\varepsilon) d$ containing exponentially many codewords in $C$) and an additional property that linearly independent minimal weight codewords do not overlap much in support (which always holds over finite fields), then the reduction in \cite{AK14} can be simplified (in particular bypassing the PCP theorem) as well as generalized to larger fields (including the reals). However, \cite{BL24} uses a random base code, and therefore the reduction is randomized. 

In this work, we derandomize (and further simplify) the reduction in \cite{BL24} for the special case of finite fields. We show that the tensoring of any \emph{$\varepsilon$-balanced 
code} (all of whose nonzero codewords have Hamming weights within a $(1 \pm \varepsilon)$ factor of each other) can be used to give a gap-producing reduction for MDP from satisfiability of homogeneous quadratic equations. 
We choose the base code to be a low rate Reed-Solomon code concatenated with a Hadamard code (which is $\varepsilon$-balanced). This gives a new deterministic proof of the hardness of MDP, that is simpler than \cite{AK14} in two ways; first the PCP theorem is not used, and secondly, the proof of the code construction is elementary whereas 
in \cite{AK14} the proof (for fields of size $\geq 3$) uses Viola’s~\cite{Viola09} construction of a pseudorandom generator for low degree polynomials.

A recent work \cite{BHIR24} also achieves constant inapproximability of NCP without using the PCP theorem (using multiplication codes). They even manage at most constant factor blow-up on the parameters when reducing from 3SAT, which leads to an exponential runtime lower bound for NCP (under the Exponential Time Hypothesis). Our reduction is nevertheless interesting due to its simplicity and also since it works essentially verbatim for MDP.

\section{Preliminaries}
\label{sec:pre}

\subsection{Homogeneous Quadratic Equations}
Any homogeneous quadratic polynomial $p:\mathbb{F}_q^n\to\mathbb{F}_q$ may be written in the form $p(x) = \sum_{i,j\in [n]}Q[i,j]x_ix_j$ for some coefficient matrix $Q=[Q[i,j]]_{i,j\in [n]} \in \mathbb{F}_q^{n\times n}$. For $X\in \mathbb{F}_q^{n\times n}$, let $Q(X):= \sum_{i,j\in [n]}Q[i,j] X[i,j]$, 
so that $p(x) = Q(xx^T)$ for all $x\in\mathbb{F}_q^n$. 
It will be convenient for us to encode homogeneous quadratics by their coefficient matrix Q. 

To prove APX-hardness of MDP, we start with the
NP-hardness of determining whether a system of homogeneous quadratic 
equations is satisfiable:
\begin{proposition}[Homogeneous Quadratic Equations Hardness]~\\
\label{prop:quadeq}
Fix any finite field $\mathbb F_q$. Given $n$ variables $x_1,\ldots,x_n$ and $m$ quadratic equations \\ $Q_1(xx^T)=0,\dots , Q_m(xx^T)=0$, it is NP-hard to distinguish between the following two cases:
    \begin{itemize}
    \itemsep=0ex
        \item (YES) There exists a non-zero vector $x \in \mathbb \{0,1\}^n$ satisfying all $m$ equations.
        \item (NO) There is no non-zero $x \in \mathbb F_q^n$ satisfying all $m$ equations.
    \end{itemize}
Furthermore, the hardness holds with the distinguished-coordinate property, i.e., in the YES case, there is a non-zero vector $x$ satisfying all equations as well as $x_n=1$.
\end{proposition}

Note that above the completeness guarantees a solution $x \in \{0,1\}^n$ whereas the soundness rules out $x \in \mathbb F_q^n$---the above is thus a promise problem. 
\Cref{prop:quadeq} is proved via reduction from the circuit satisfiability (\textsc{Circuit-SAT}) problem. The proof follows the standard template for exact NP-Hardness results, and we defer it to \Cref{sec:quadeq}.

\subsection{Tensor Codes and Distance Amplification}

We show hardness of MDP by first generating a constant factor gap and then using the standard observation that the minimum distance of a code is multiplicative under the usual tensor product operation, which we prove below for completeness.

The tensor product of two subspaces $U\subseteq \mathbb{F}_q^n$ and $V\subseteq \mathbb{F}_q^m$, denoted by $U\otimes V$ may be defined as the space of matrices $M\in \mathbb{F}_q^{n \times m}$ such that every row of $M$ lies in $V$ and every column of $M$ lies in $U$. 

Let $\|U\|_0$ denote the minimum distance, i.e., 
$\|U\|_0 := \min_{u\in U\setminus \{0\}} \|u\|_0$. 
Then we have 
\begin{observation}
\label{mult:dist}
    For any subspaces $U\subseteq \mathbb{F}_q^n$ and $V\subseteq \mathbb{F}_q^m$,
    $\|U\otimes V\|_0 = \|U\|_0 \cdot \|V\|_0$. 
\end{observation}
\begin{proof}
    The LHS is at most the RHS since we may consider the element $uv^T\in U\otimes V$, 
    where we choose $u\in U$ that has sparsity $\|U\|_0$ and $v\in V$ that has sparsity $\|V\|_0$. 

    For the other direction, consider any nonzero $M\in U\otimes V$. 
    There must be some non-zero entry in $M$, and so there is at least one nonzero column. Since this column lies in $U$, it must have at least $\|U\|_0$ nonzero 
    entries, and therefore at least $\|U\|_0$ rows of $M$ are nonzero. Each such row lies in $V$ and hence has $\|V\|_0$ nonzero entries. We conclude that $M$ has at least $\|U\|_0\cdot \|V\|_0$ nonzero entries. 
\end{proof}
Applying the above observation inductively, we have:
\begin{observation}
\label{tensoring}
    For any subspaces $U\subset \mathbb{F}_q^n$ and any $t\in \mathbb{N}$,
    $\|U^{\otimes t}\|_0 = \|U\|_0^t $. 
\end{observation}

\section{Rank testing via tensor codes}

The gap-producing randomized reduction in \cite{BL24} proceeds by first observing that testing (nonzero) satisfiability of a system of quadratic equations 
$Q_1(xx^T)=0,\dots , Q_m(xx^T)=0$ is 
equivalent to testing whether or not the subspace of matrices given by 
$\{X~:~X=X^T,~Q_1(X)=0,\dots , Q_m(X)=0\}$ contains a rank-1 matrix. 
Then \cite{BL24} shows how tensoring of a base code with appropriate structure 
can be used to test via Hamming weight if a matrix subspace contains a 
rank-1 matrix. 
We follow the same template, but our approach is even simpler as over finite fields, we are able to embed the above subspace directly inside a tensor code. We begin by discussing the simple but crucial properties of the Hamming weight of matrices of rank $> 1$ in a tensor code.

\paragraph{Hamming Weight of Rank $\geq 2$ Elements of a Tensor Code.}
Recall that in \Cref{mult:dist}, the upper bound $\|U\|_0^2$ on $\|U\otimes U\|_0$ (taking $V=U$) is attained by a rank-1 matrix. 
Implicit in \cite{AK14} is the insightful observation that codewords of rank $\geq 2$ in $U\otimes U$ have Hamming weight at least $(1+\frac{1}{q})\|U\|_0^2$ (over a finite field $\mathbb{F}_q$), which is significantly larger than the minimum distance. 
We include the proof here for completeness. 
\begin{lemma}
\label{rk2hw}
    For any subspace $U\subset \mathbb{F}_q^n$ and any $M\in U\otimes U$ of rank at least $2$, we have $\|M\|_0 \geq (1+\frac{1}{q})\cdot \|U\|_0^2$. 
\end{lemma}

\begin{proof}
    Since $M$ is of rank at least two, there are two linearly independent 
    columns, that  by \Cref{lem:overlap} (proved below), have joint support of size at least $(1+\frac{1}{q}) d$. Thus at least $(1+\frac{1}{q}) d$ rows of $M$ are non-zero, and since they lie inside $U$, each of these rows has at least $\|U\|_0$ non-zero entries. Thus $\|M\|_0\ge (1+\frac{1}{q}) \|U\|_0^2$.
\end{proof}

\begin{fact}[Fact 2.7 in~\cite{AK14}]\label{lem:overlap}
        Let $U \subseteq \mathbb F_q^n$ be a subspace. Then for any two linearly independent vectors $x,y \in \mathbb F_q^n$, we have $|\text{supp}(x) \cup \text{supp}(y)|\ge (1+\frac{1}{q})\|U\|_0$.
    \end{fact}

    \begin{proof}
        Let $m$ be the number of coordinates such that $x_i \neq 0$ or $y_i \neq 0$ but not both, and let $m'$ be the number of coordinates such that $x_i \neq 0$ and $y_i \neq 0$. Clearly,
        $$m+2m' \ge 2d.$$
        Since there are only $q-1$ choices of values for $x_i/y_i$ (for $x_i,y_i\neq 0$), there must exist $\lambda \neq 0$ 
        so that the vector $x -\lambda y$ has at most $m+m'-\frac{m'}{q-1}$ non-zero entries. This implies
        $$m+m'-\frac{m'}{q-1}\ge d.$$
        Multiplying the first inequality by $\frac{1}{q}$, the second by $\frac{q-1}{q}$, and adding, gives $m+m' \ge (1+\frac{1}{q})d$ as desired.
    \end{proof}

\paragraph{$\varepsilon$-Balanced Codes.}
At the heart of our reduction that generates constant factor hardness for MDP is the observation that the above connection 
between rank and Hamming weight can be made two-sided, assuming the base 
code being tensored is $\varepsilon$-balanced. More specifically, we observe that any codeword in $C\otimes C$ for an $\varepsilon$-balanced code $C$ has 
near-minimum Hamming weight if and only if it is rank-1. 
\begin{definition}\label{def:eps-balanced}
    For any constant $\varepsilon>0$, we say a linear error-correcting code with encoding map $C:\mathbb F_q^n \to \mathbb F_q^N$ with distance $d$ is $\varepsilon$-balanced\footnote{
    The usual definition of $\varepsilon$-balanced is for binary linear codes and has 
    the additional requirement that the minimum distance is $N/2(1-\Theta(\varepsilon))$. 
    For our purposes, the minimum distance is unconstrained. We abuse terminology 
    and continue to use the term $\varepsilon$-balanced. 
    We also use this terminology for larger fields. 
    } 
    if the Hamming weight of every nonzero codeword lies in the range $[d,(1+\varepsilon)d]$.
\end{definition}
We remark that $\varepsilon$-balanced codes of not-too-small dimension satisfy a weak version of local density, namely a Hamming ball of radius $(1+\varepsilon)d$ contains exponentially many (in fact all) codewords. 

It turns out that 
$\varepsilon$-balanced codes can be easily constructed by concatenating a Reed-Solomon code with the Hadamard code \cite{AGHP92}. Specifically, we have the following lemma:
\begin{lemma} \label{lem:eps-biased}
    For any constant $\varepsilon>0$, any finite field $\mathbb F_q$, and any $n \in \mathbb N$, there exists $N \le (qn/\varepsilon)^2$ and a linear code $C:\mathbb F_q^n \to \mathbb F_q^{N}$ with minimum distance at least $d = (1-\varepsilon)(1-\frac{1}{q})N$, satisfying
    $$\|C(x)\|_0 \in \left[(1-\varepsilon)\left(1-\frac{1}{q}\right)N,\left(1-\frac{1}{q}\right)N\right], \forall x \in \mathbb F_q^n \setminus \{0\} .$$
\end{lemma}

Note that when $\varepsilon<\frac{1}{2}$, we have $\left[(1-\varepsilon)\left(1-\frac{1}{q}\right)N,\left(1-\frac{1}{q}\right)N\right] \subseteq [d,(1+2\varepsilon) d]$, which means the code is $(2\varepsilon)$-balanced.

\begin{proof}[Proof of \Cref{lem:eps-biased}]
Pick $m$ to be the smallest integer so that $n \le \varepsilon q^m$. Let $Q = q^m$. Note that $Q \le q n/\varepsilon$.
   Let $\text{RS} : \mathbb F_q^n \to \mathbb F_Q^Q$ be a Reed-Solomon encoding map that maps polynomials of degree $< n$ over $\mathbb F_q$ to their evaluations at all points in the extension field $\mathbb F_Q$. Now concatenate this encoding with the Hadamard encoding that maps $\mathbb F_Q$, viewed as vectors in $\mathbb F_q^m$ under some canonical basis, to $\mathbb F_q^{q^m}$. The resulting concatenated code has block length $N=Q \cdot q^m = Q^2  \le ( q n/\varepsilon)^2$. 
   
   The distance of the concatenated code is at least $(1-\varepsilon)(1-\frac{1}{q}) N$, since the Reed-Solomon code has distance greater than $Q - n \ge (1-\varepsilon) Q$ and the Hadamard code has distance $(1-\frac{1}{q})Q$. The lower bound on the weight of every nonzero codeword follows from distance, while the upper bound comes from the fact that each of the $Q$ symbols in Reed-Solomon code contributes at most $(1-\frac{1}{q})Q$ Hamming weight after encoding by the Hadamard code. 
\end{proof}

\paragraph{Rank-1 Testing via Hamming Weight.}
Equipped with such a code, we obtain the key observation 
regarding rank-1 testing that underlies our MDP reduction:
\begin{lemma}[Rank-1 Testing via Hamming Weight of a Linear Transformation]~\\
\label{rk1test}
    Let $C\in \mathbb{F}_q^{N\times n}$ be a generator matrix of an $\varepsilon$-balanced code of minimum distance $d \ge 1$. 
    Let $X\in \mathbb{F}_q^{n\times n}$. Then 
    \begin{enumerate}
        \item[-] if $X$ is rank one,  $\|CXC^T\|_0\leq (1+\varepsilon)^2\cdot d^2$ 

        \item[-] if $X$ has rank at least two, $\|CXC^T\|_0\geq (1+\frac{1}{q})\cdot d^2$. 
    \end{enumerate}
\end{lemma}

\begin{proof}
    The first claim follows  
    since if $X$ is rank one, $CXC^T=(Cu)(Cv)^T$ for some $u,v\in \mathbb{F}_q^n$. Then by the definition of $\varepsilon$-balanced, 
    $\|CXC^T\|_0 = \|Cu\|_0\cdot \|Cv\|_0\leq (1+\varepsilon)^2\cdot d^2$. 

    For the second claim, note that if $X$ has rank at least 2, then $CXC^T$ has rank at least 2 since $C$, being the generator matrix of a code of positive distance, has full column rank. Furthermore the rows and columns of $CXC^T$ all lie in the 
    code whose generator matrix is $C$. The second claim then follows from \Cref{rk2hw}. 
\end{proof}

\section{Inapproximability of MDP}
\label{sec:mdp}

We now prove \Cref{thm:mdp} by presenting a gap-producing reduction from homogeneous quadratic equations to MDP.

\subsection{Reduction }
\label{subsec:mdp_red}
\paragraph{Input:} A parameter $\varepsilon>0$ and a system of homogeneous quadratic equations of the form 
\begin{align}
\label{hq}
Q_1(xx^T)=0,\dots , Q_m(xx^T)=0
\end{align}

\paragraph{Output Subspace:}
Let $\varepsilon = \frac{1}{9q}$, and let $C\in \mathbb{F}_q^{N\times n}$ be the generator matrix of an $\varepsilon$-balanced code of minimum distance $d$. Our output subspace 
is defined as 
\begin{equation}
V:= \{CXC^T :~ Q_1(X)=0,\dots , Q_m(X)=0,~
X^T=X,~X\in\mathbb{F}_q^{n\times n}\}
\label{hred}
\end{equation}

Using the construction of an $\varepsilon$-balanced code from \Cref{lem:eps-biased}, we reduce an instance of homogeneous quadratic equations with $n$ variables to an MDP instance with $N^2=\text{poly}(n,\frac{1}{\varepsilon})$ variables. A basis of $V$ can be computed in polynomial time by considering 
the basis $\{CXC^T:X\in B\}$, where $B$ is a basis of 
$\{X:Q_1(X)=0,\ldots ,Q_m(X)=0,~X^T=X\}$.

\subsection{Analysis}
\paragraph{Completeness.} Let $x \in \mathbb \{0,1\}^{n}$ be a non-zero solution to the system \Cref{hq}. Then $(Cx)(Cx)^T\in V\setminus\{0\}$ and satisfies $\|(Cx)(Cx)^T\|_0 \le (1+\varepsilon)^2 d^2 \leq (1+\frac{1}{3q})d^2$. 

\paragraph{Soundness.} Suppose there is no non-zero solution to system \Cref{hq}, we argue that 
$\|V\|_0 \ge (1+\frac{1}{q}) d^2$. Consider any non-zero $Y\in V$, and let 
$X$ be such that $Y=CXC^T$. 

If $X$ has rank at least 2, by \Cref{rk1test}, $\|Y\|_0\ge (1+\frac{1}{q}) d^2$. So it remains to consider the case where $X$ has rank 1. Since $X$ is symmetric, we conclude $X = xx^T$ for some non-zero $x\in \mathbb F_q^n$, which implies that for every $\ell \in [m]$, $Q_\ell(xx^T)=0$,
and thus $x$ is a solution to the system \Cref{hq}, contradicting our assumption.

This yields NP-Hardness of approximating MDP within a factor of 
$(1+\frac{1}{q})/(1+\frac{1}{3q})= 1+\frac{2}{3q+1}$. 
By simple tensoring, one can increase the gap to any constant, with only a polynomial blow-up on the instance size, and to almost-polynomial gap with a quasi polynomial blow-up in instance size. This completes the proof of \Cref{thm:mdp}.

\section{Inapproximability of NCP}
\label{sec:ncp}

In this section we will prove \Cref{thm:ncp} using a straightforward gap-preserving reduction from MDP. Towards this end, we fist note that the MDP 
reduction shown above can be easily modified to have the distinguished-coordinate property. We defer the proof to the appendix. 

\begin{proposition}[MDP Hardness with a Distinguished Coordinate]~\\
\label{dist:mdp}
For any finite field $\mathbb F_q$ and any constant $c>1$ (resp. for any constant $\varepsilon>0$), it is NP-hard (resp. Quasi-NP-Hard) 
given an input subspace $V\subseteq \mathbb{F}_q^N$ to distinguish between the following cases: 
    \begin{itemize}
        \item (YES) There exists $x \in V\setminus \{0\}$ satisfying $\|x\|_0 \le k,~x_N=1$;
        \item (NO) For all $x \in V\setminus \{0\}$, $\|x\|_0 > c \cdot k$ (resp. $\|x\|_0 > 2^{\log^{1-\varepsilon}N} \cdot k$).
    \end{itemize}
\end{proposition}

\begin{proof}[Proof of \Cref{thm:ncp}]
    Consider the reduction from MDP (with a distinguished coordinate) 
    to NCP, given by mapping a subspace $V \subseteq \mathbb F_q^{n}$ to the affine subspace $V':=\{x\in V:x_n=1\}$. 
    The claim then follows from \Cref{dist:mdp}. 

    Completeness follows from the distinguished-coordinate property of MDP.

    For soundness, note that any $x\in V'$ is a non-zero vector in $V$,  with the same sparsity. 
\end{proof}

\begin{remark}
A gap-preserving Cook-reduction from MDP to NCP was known in the literature (see e.g., \cite{GMSS99} and Exercise 23.13 in \cite{GRS18}). However, we have a Karp-reduction here thanks to the distinguished-coordinate property.

We also remark that there are non-trivial ways to amplify gap for NCP (see e.g., Theorem 22 of \cite{LLL24} and Section 4.2 of \cite{BGKM18}). Thus one can start with the NP-hardness of non-homogeneous quadratic equations, perform the reduction in \Cref{sec:mdp} to get a mild constant gap for NCP, and then amplify the gap. This is another way to prove \Cref{thm:ncp}.
\end{remark}

\begin{remark}    
    The relatively near codeword problem parameterized by $\rho\in (0,\infty)$ (denoted by $\text{RNC}^{\,(\rho)}$) is a promise problem 
    defined as follows. Given a subspace $V\subseteq \mathbb{F}_q^n$ of an (unknown) minimum distance $d$, a vector $b\in \mathbb{F}_q^n$ and an integer $t$ with the promise that $t<\rho\cdot d$, the task is to find a codeword in $V$ of hamming distance at most $t$ from $b$. 

    Our reduction for NCP yields NP-Hardness of $\text{RNC}^{\,(1+\varepsilon)}$ 
    for any fixed $\varepsilon>0$ (it also yields hardness for the gap version of RNC for an appropriately small constant gap). \cite{DMS93} show 
    hardness of $\text{RNC}^{\,(\rho)}$ for any $\rho>1/2$, albeit under randomized reductions. 
\end{remark}

\section*{Acknowledgement}
\noindent Xuandi would like to thank Yican Sun for helpful comments and discussions.

\bibliographystyle{alpha}
\bibliography{main}

\appendix

\section{Proof of \Cref{prop:quadeq}}
\label{sec:quadeq}

We reduce circuit satisfiability (\textsc{Circuit-SAT}) problem to the homogeneous quadratic equations problem. \textsc{Circuit-SAT} asks whether a Boolean circuit, which consists of input gates as well as AND, OR, NOT gates with fan-in (at most) two and fan-out unbounded, has an assignment of its input that makes output true. It is a prototypical NP-complete problem, since the Cook-Levin theorem is sometimes proved on \textsc{Circuit-SAT} instead of \textsc{3SAT} (see e.g., Lemma 6.10 in \cite{AB09}). 

\begin{proof}[Proof of \Cref{prop:quadeq}]
    Given a \textsc{Circuit-SAT} instance $C$ with $n$ gates $y_1,\ldots,y_n$ (including input gates and logic gates), we build $n+1$ variables $\{x_1,\ldots,x_n,z\}$, and add the following equations:
    \begin{itemize}
        \item $x_i(x_i-z)=0, \forall i \in[n]$;
        \item for each AND gate $y_k = y_i \land y_j$ in $C$, an equation $x_k^2 = x_i x_j$;
        \item for each OR gate $y_k = y_i \lor y_j$ in $C$, an equation $z^2 - x_k^2 = (z-x_i) (z-x_j)$;
        \item for each NOT gate $y_k = \lnot y_i$ in $C$, an equation $z^2 - x_k^2 = x_i^2$;
        \item for $y_k$ being the output gate, an equation $z^2 = x_k^2$.
    \end{itemize}

    \paragraph{Completeness.}
    Let $y_1,\ldots,y_n \in \{0,1\}$ be an assignment to the gates of $C$ that makes the output true. It's easy to verify $$\left\{\begin{aligned}
        z & = 1, \\
        x_i & = y_i, &  \forall i \in [n]
    \end{aligned}\right.$$
    is a solution to the system of quadratic equations.

    \paragraph{Soundness.}
    If $z=0$, then by the first set of equations, each $x_i$ has to be 0 and this is an all-zero solution. Otherwise, setting each $y_i=x_i z^{-1}$ is a satisfying assignment of $C$ since it satisfies every gate in $C$ and ensures that the output gate is true.

    The distinguished-coordinate property follows by setting $z$ to be the last variable.
\end{proof}

\section{MDP Hardness with a Distinguished Coordinate}
\begin{proof}[Proof of \Cref{dist:mdp}]
    We proceed very similarly to \Cref{thm:mdp}. 
    We reduce from homogeneous 
    quadratic equations with a distinguished coordinate (\Cref{prop:quadeq}), and append a distinguished coordinate to the reduction in \Cref{thm:mdp}: 
    \begin{equation}
        V:= \{(CXC^T,X_{n,n}) :~ Q_1(X)=0,\dots , Q_m(X)=0,~
        X^T=X,~X\in\mathbb{F}_q^{n\times n}\}
    \label{hred}
    \end{equation}
    For completeness, let $x \in \mathbb \{0,1\}^{n}$ be a non-zero solution to the system \Cref{hq} that satisfies $x_n=1$. Then $((Cx)(Cx)^T,1)\in V$ and has Hamming weight at most 
    $(1+\frac{1}{3q})d^2+1$. 

    For soundness, note that if $(CXC^T,X_{n,n})$ is non-zero, then 
    $X$ must be non-zero. 
    The remainder of the analysis proceeds 
    identically to the soundness analysis in \Cref{thm:mdp}, and we conclude $\|V\|_0 \geq (1+\frac{1}{q})\cdot d^2$. 

    Finally we amplify the gap using tensoring (\Cref{tensoring}), 
    and we note that the distinguished coordinate property of the YES case is preserved under tensoring. 
\end{proof}
\end{document}